\documentclass{article}

\usepackage{arxiv}

\usepackage[utf8]{inputenc} 
\usepackage[T1]{fontenc}    
\usepackage{hyperref}       
\usepackage{url}            
\usepackage{booktabs}       
\usepackage{amsfonts}       
\usepackage{nicefrac}       
\usepackage{microtype}      
\usepackage{graphicx}
\usepackage{natbib}
\usepackage{doi}
\usepackage{physics,amsmath}
\usepackage{authblk}
\usepackage{amsthm}
\usepackage{placeins}
\usepackage{soul}
\usepackage{xcolor}
\usepackage{float}

\newtheorem{theorem}{Theorem}[section]
\newtheorem{lemma}[theorem]{Lemma}

\newcommand{\average}[1]{\ensuremath{\{\!\{#1\}\!\}} }
\newcommand{\jump}[1]{\ensuremath{[\![#1]\!]} }

\title{Thermodynamic consistency and structure-preservation in summation by parts methods for the moist compressible Euler equations}

\author[1]{Kieran Ricardo}
\author[2]{David Lee}
\author[1,3]{Kenneth Duru}
\affil[1]{Mathematical Sciences Institute, Australian National University, Canberra, Australia}
\affil[2]{Bureau of Meteorology, Melbourne, Australia}
\affil[3]{Department of Mathematical Sciences , University of Texas at El Paso, USA}
\date{} 					

\graphicspath{ {./plots} }

\begin{document}
\maketitle

\begin{abstract}
Moist thermodynamics is a fundamental driver of atmospheric dynamics across all scales, making accurate modeling of these processes essential for reliable weather forecasts and climate change projections. However, atmospheric models often make a variety of inconsistent approximations in representing moist thermodynamics. These inconsistencies can introduce spurious sources and sinks of energy, potentially compromising the integrity of the models.

Here, we present a thermodynamically consistent and structure preserving formulation of the moist compressible Euler equations. When discretised with  a summation by parts method, our spatial discretisation conserves: mass, water, entropy, and energy. These properties are achieved by discretising a skew symmetric form of the moist compressible Euler equations, using entropy as a prognostic variable, and the summation-by-parts property of discrete derivative operators. Additionally, we derive a discontinuous Galerkin spectral element method with energy and tracer variance stable numerical fluxes, and experimentally verify our theoretical results through numerical simulations.
\end{abstract}


\section{Introduction}

Atmospheric models typically utilise multiple different, and often inconsistent, thermodynamic approximations throughout a single code. These thermodynamic inconsistencies introduce resolution independent spurious sources and sinks of energy and hence violate the first law of thermodynamics. For climate models these energy errors can be significant, and are believed to adversely affect the integrity of long-time climate simulations 
\cite{lauritzen2018ncar, eldred2022thermodynamically}.

To improve the energy budgets of global climate models, recent studies have explored the application of thermodynamic potentials to ensure thermodynamic consistency \cite{thuburn2017use, staniforth2019forms, eldred2022thermodynamically}. This approach focuses on approximating a single thermodynamic quantity—known as the thermodynamic potential—and deriving all other relevant quantities from it. The most commonly employed potentials in this context are internal energy and the Gibbs function. For a comprehensive overview of thermodynamic potentials we  refer the reader to \cite{staniforth2019forms} and \cite{eldred2022thermodynamically}.

The thermodynamic potential approach was first proposed in \cite{thuburn2017use}, which introduces a thermodynamically consistent semi-Lagrangian model for a two-phase liquid-vapor system utilizing the Gibbs potential as a function of its natural variables: pressure and temperature. However, extending this method to include a third ice phase poses challenges, as the Gibbs potential does not uniquely define an equilibrium state at the triple point.

To address this issue, \cite{bowen2022consistent1, bowen2022consistent2} explore an alternative approach that employs internal energy as the potential in a semi-Lagrangian discretization of the three-phase moist compressible Euler equations, effectively sidestepping the triple point ambiguity associated with the Gibbs potential. Similarly, \cite{guba2024energy} presents a thermodynamically consistent method that also uses internal energy as the potential, implemented within a horizontally semi-Lagrangian and vertically spectral element discretization. While these methods ensure thermodynamic consistency and therefore discretize an energy-conserving set of equations, it is important to note that the discretizations themselves may not guarantee energy conservation or stability.

To the best of their knowledge, the authors are not aware of any mimetic or structure-preserving spatial discretizations for the moist compressible Euler equations, however it is noteworthy that significant progress has been made in developing such methods for the dry compressible Euler equations, as well as for related shallow water and thermal shallow water equations \cite{mcrae2014energy, lee2018discrete, lee2020mixed, taylor2010compatible, eldred2019quasi, natale2016compatible, cotter2023compatible}. These advancements encompass a variety of approaches, including mixed finite element, continuous Galerkin, and discontinuous Galerkin methods. Although the specifics differ among these techniques and equations, they all maintain structure preservation by discretising the equations in a vector-invariant skew-symmetric form and utilizing summation-by-parts (SBP) operators.

Structure preservation may reduce model biases and more accurately represent physical processes, but it does not guarantee numerical stability for compressible flows with active tracers, such as the dry and moist compressible Euler equations and the thermal shallow water equations. We hypothesize that this limitation arises because energy is not a mathematical entropy for these equations. In this context, mathematical entropy refers to any conserved quantity that is a convex function of the prognostic variables; for instance, energy in the shallow water equations and physical entropy in the conservative dry Euler equations. Numerous studies \cite{waruszewski2022entropy, ranocha2020entropy, ducros2000high, morinishi2010skew, sjogreen2019entropy, gassner2016split, hennemann2021provably, pirozzoli2010generalized} have demonstrated that discrete mathematical entropy stability is both necessary and often sufficient for ensuring numerical stability. These methods often utilize the conservative form of the equations and aim to conserve the physical entropy by employing techniques such as splitting derivative operators, SBP, and entropy stable numerical fluxes. 

In \cite{ricardo2024dg, ricardo2024entropy}, the authors introduce numerical methods for the thermal shallow water
and dry compressible Euler equations that are both mimetic and mathematical entropy stable. These studies identify tracer
variance as a form of mathematical entropy and propose a split-form, skew-symmetric formulation of the equations. The resulting formulations are discretized using SBP operators, ensuring both entropy and energy stability.

In this work, we extend this structure-preserving and mathematical entropy-stable approach to the moist compressible Euler equations. The result is a thermodynamically consistent skew-symmetric formulation that, when discretized using a carefully designed SBP scheme, maintains both structure preservation and mathematical entropy stability. We then develop a novel energy and mathematical entropy stable discontinuous Galerkin spectral element method (DG-SEM) for these equations.

The rest of the paper is as follows. In section 2 we derive an operator-split skew-symmetric formulation of the moist compressible Euler equations. We demonstrate that energy and mathematical entropy conservation can be proved using only integration by parts, and therefore any SBP discretisation of these equations will inherit discrete analogues of these conservation properties. In section 3 we outline the thermodynamic approximation we use in our numerical method. The discrete conservation and stability results are independent of the particular thermodyanmic approximation used but we include this for completeness. In section 4 we introduce DG-SEM and detail our novel DG-SEM discretisation. Following this, in section 5 we prove that our discrete method: conserves mass, water, and entropy; and is semi-discretely energy and (mathematical) entropy stable. These theoretical results are then verified via numerical experiments in section 6. Finally, our conclusions from this study are presented in section 7.

\section{Continuous moist compressible Euler equations}

In this section we derive a split-form skew-symmetric formulation of the vector invariant moist compressible Euler equations and subsequently demonstrate that the moist compressible Euler equations conserve energy and tracer variance. This provides a set of equations that can be discretised using any SBP operators to yield a semi-discretely energy and tracer-variance conserving scheme. While much of this analysis exists in the literature \cite{bannon2003hamiltonian} we include it here to aid the exposition, and the analysis performed in this paper.

\subsection{Vector invariant pressure gradient formulation}
For simplicity we consider the 2D moist compressible Euler equations in thermodynamic equilibrium on domain $\Omega$ with periodic boundaries in the horizontal direction and wall/no-flux boundaries in the vertical. Our analysis also applies to the 3D moist compressible Euler equations and can be extended to include non-equilibrium thermodynamics \cite{bowen2022consistent2}, noting that this introduces source terms into the tracer variance evolution equation at the continuous level. The vector invariant moist compressible Euler equations in pressure gradient form are
\begin{equation}
	\frac{\partial \vb{v}}{\partial t} + \boldsymbol{\omega}\cross \vb{v} + \nabla \left(\tfrac{1}{2}|\vb{v}|^2 + \Phi\right) +\frac{1}{\rho}\nabla p = 0,
	\label{eq:vec-vel}
\end{equation}
\begin{equation}
	\frac{\partial \rho}{\partial t} + \nabla\cdot \rho\vb{v} = 0,
 \label{eq:vec-density}
\end{equation}
\begin{equation}
	\frac{\partial \rho\eta}{\partial t} + \nabla\cdot \eta\rho\vb{v} = 0,
 \label{eq:vec-entropy}
\end{equation}
\begin{equation}
	\frac{\partial \rho q_w}{\partial t} + \nabla\cdot q_w\rho\vb{v} = 0,
 \label{eq:vec-water}
\end{equation}
subject to the boundary condition
\begin{equation}
    \rho\vb{v}\left(\vb{x}, t\right)\cdot \vb{n} = 0, \vb{x}\in\partial \Omega,
 \label{eq:no-flux-boundary}
\end{equation}
and the initial condition
\begin{align}
    \vb{v} = \vb{v}_0(\mathbf{x}), \quad {\rho} = {\rho}_0(\mathbf{x}), \quad {\eta} = {\eta}_0(\mathbf{x}), \quad {q_w} = {q_w}_0(\mathbf{x}), \quad \mathbf{x} \in \Omega.
\end{align}
Here our prognostic variables are: flow velocity $\vb{v}$, density of the water-air mixture $\rho$, entropy $\rho \eta$, and total water $\rho q_w$. Additionally, $\eta$ is the specific entropy, $q_w$ is the total water mass fraction, $\boldsymbol{\omega}=\nabla \cross \vb{v} + f$ is the absolute vorticity, $f$ is the Coriolis parameter, $\Phi=gz$ is the gravitational potential, $\vb{n}$ is the unit normal of the boundary, and $p(\rho, \eta, q_w)$ is the pressure. These equations are closed with an approximate expression for the pressure, which if not chosen in a thermodynamically consistent manner can violate energy conservation.

\subsection{Thermodynamics and energy conservation}

In this section we follow a similar analysis to \cite{bannon2003hamiltonian} to derive a skew symmetric form of the moist compressible Euler equations and subsequently show that energy is conserved. The energy density $e$ of the moist compressible Euler equations is defined as the sum of kinetic, potential, and internal energy densities
\begin{equation}
    e = \tfrac{1}{2}\rho\norm{\vb{v}}^2 + \rho \Phi + \rho u,
\end{equation}
where $u$ is the specific internal energy. Energy conservation follows from the fundamental thermodynamic identities 
\begin{equation}
	\left(\frac{\delta  u}{\delta  \rho}\right)_{\eta, q_w}=\frac{p}{\rho^2},\; \left(\frac{\delta  u}{\delta  \eta}\right)_{\rho, q_w}=T, \left(\frac{\delta u}{\delta  q_w}\right)_{\rho, \eta}=\mu_w,
	\label{eq:thermo-derivs}
\end{equation}
where $T$ is the temperature and $\mu_w$ is the chemical potential of water. This enables us to calculate the functional derivatives of energy w.r.t. the prognostic variables
\begin{equation}
	\frac{\delta e}{\delta \vb{v}} = \rho \vb{v} := \vb{F}, \; \frac{\delta e}{\delta \rho\eta} = T,\; \frac{\delta e}{\delta \rho q_w} = \mu_{w} \;,
 \label{eq:func-derivs-1}
\end{equation}
\begin{equation}
	\frac{\delta e}{\delta \rho} = \tfrac{1}{2}|\vb{v}|^2 + \Phi + h - T\eta - \mu_w q_w := G,
 \label{eq:func-derivs-2}
\end{equation}
where $h=u + \frac{p}{\rho}$ is the enthalpy. Using the identity $\nabla u = \frac{\delta u}{\delta\rho}\nabla\rho + \frac{\delta u}{\delta\eta}\nabla\eta + \frac{\delta u}{\delta q}\nabla q$ we express the pressure gradient term as
\begin{equation}
\begin{split}
    \frac{1}{\rho}\nabla p &= \nabla h - \nabla u + \frac{p}{\rho^2}\nabla \rho = \nabla h - T\nabla \eta - \mu_w\nabla q_w \\
    &= \nabla \left(h - T\eta - q_w\mu_w\right) + \eta \nabla T + q_w \nabla \mu_w.
    \end{split}\label{eq:pgrad}
\end{equation}
Substituting \eqref{eq:pgrad} into the velocity equation \eqref{eq:vec-vel} yields a skew-symmetric formulation of the moist compressible Euler equations
\begin{equation}
	\frac{\partial \vb{v}}{\partial t} + \boldsymbol{\omega}\cross \vb{v} + \nabla G + \eta \nabla T + q_w \nabla \mu_w = 0,
	\label{eq:skew-vel1}
\end{equation}
\begin{equation}
	\frac{\partial \rho}{\partial t} + \nabla\cdot \rho\vb{v} = 0,
 \label{eq:skew-density}
\end{equation}
\begin{equation}
	\frac{\partial \rho\eta}{\partial t} + \nabla\cdot \eta\rho\vb{v} = 0,
 \label{eq:skew-entropy}
\end{equation}
\begin{equation}
	\frac{\partial \rho q_w}{\partial t} + \nabla\cdot q_w\rho\vb{v} = 0.
 \label{eq:skew-water}
\end{equation}
We now prove that the moist compressible Euler equations conserve energy by starting with the skew-symmetric form and using integration by parts.

\begin{theorem}\label{theo:conservation_of_energy}
Consider the the moist compressible Euler equations \eqref{eq:skew-vel1}-\eqref{eq:skew-water} on the smooth spatial domain $\Omega \subset \mathbb{R}^2$ with periodic and wall \eqref{eq:no-flux-boundary} boundary conditions. For continuous solutions $(\vb{v}, \rho, \eta, q_w): \Omega\times\mathbb{R}^{+} \to \mathbb{R}^{5}$, at time $t \geq 0$ let the total energy be denoted by $\mathcal{E}(t) = \int_\Omega{ed\Omega}$. The total energy is conserved, that is we have
\begin{align}
    \frac{d \mathcal{E}}{dt} = 0, \quad \forall t\ge 0.
\end{align}
\end{theorem}
\begin{proof}
We consider the time derivative of the total energy and apply the chain rule in time, we have
\begin{align}
     \frac{d\mathcal{E}}{dt} = \int_{\Omega}{\frac{\partial e}{\partial t} d\Omega} &= \int_{\Omega}{\left(
            \frac{\delta e}{\delta \vb{v}}\cdot \frac{\partial \vb{v}}{\partial t}
            + \frac{\delta e}{\delta \rho} \frac{\partial \rho}{\partial t}
            + \frac{\delta e}{\delta (\rho\eta)} \frac{\partial \left(\rho\eta\right)}{\partial t}
            + \frac{\delta e}{\delta (\rho q_w)}\frac{\partial \left(\rho q_w\right)}{\partial t}
            \right)d\Omega}.
\end{align}
Using the functional derivatives in \eqref{eq:func-derivs-1} and \eqref{eq:func-derivs-2}
    \begin{equation}
        \begin{split}
            \frac{d\mathcal{E}}{dt}  &= \int_{\Omega}{\left(
            \vb{F}\cdot \frac{\partial \vb{v}}{\partial t}
            + G \frac{\partial \rho}{\partial t}
            + T \frac{\partial \left(\rho\eta\right)}{\partial t}
            + \mu_w\frac{\partial \left(\rho q_w\right)}{\partial t}
            \right)d\Omega} \\
        &= -\int_{\Omega}{\left(\vb{F}\cdot \left(
    \boldsymbol{\omega}\cross \vb{v}\right)+
            \vb{F}\cdot \nabla G + \eta \vb{F}\cdot \nabla T
            + q_w \vb{F}\cdot \nabla \mu_w\right)d\Omega} \\
            &- \int_{\Omega}{\left(G\nabla \cdot \vb{F} + T\nabla \cdot \left(\eta\vb{F}\right) + \mu_w\nabla \cdot \left(q_w\vb{F}\right)
            \right)d\Omega}.
        \end{split}
    \end{equation}
    Noting that $\vb{F}\cdot \left(
    \boldsymbol{\omega}\cross \vb{v}\right) = 0$ and collecting terms together we have
    \begin{equation}
        \begin{split}
            \frac{d\mathcal{E}}{dt}  
        = &-\int_{\Omega}{\left(
            \left(\vb{F}\cdot \nabla G + G\nabla \cdot \vb{F}\right) + \left(\eta \vb{F}\cdot \nabla T
            + T\nabla \cdot \left(\eta\vb{F}\right)\right)\right)d\Omega} \\
            &- \int_{\Omega}{\left( q_w \vb{F}\cdot \nabla \mu_w + \mu_w\nabla \cdot \left(q_w\vb{F}\right)
            \right)d\Omega} .
        \end{split}
    \end{equation}
    Applying integration by parts we have
    \begin{equation}
        \begin{split}
            \frac{d\mathcal{E}}{dt}  
        &= -\oint_{\partial \Omega}{\left(\left(G + \eta T + q_w \mu_w \right)\vb{F}\right)\cdot \vb{n}dl}.
        \end{split}
    \end{equation}
    The final step of the proof involves eliminating the boundary integral  by imposing wall boundary conditions $\vb{F}\cdot\vb{n}=0$ and cancellation of boundary terms across periodic boundaries, giving
    \begin{equation}
        \begin{split}
            \frac{d\mathcal{E}}{dt}  
        &= 0.
        \end{split}
    \end{equation}
    The proof is complete.
    
\end{proof}
The proof only relies upon integration by parts and $\vb{F}\cdot \left(
    \boldsymbol{\omega}\cross \vb{v}\right) = 0$, implying that a well-designed SBP discretisation of \eqref{eq:skew-vel1}-\eqref{eq:skew-water} on collocated grids would hopefully conserve energy at the discrete level. Note that non-collocated methods can conserve energy by discretising the rotational term in the form $\vb{q}\cross\vb{F}$, where $q = \frac{\vb{\omega}}{\rho}$ \cite{mcrae2014energy,lee2018discrete}.

\subsection{Mathematical entropy and tracer variance}

Entropy in the mathematical sense is any quantity that is both conserved and a convex function of the prognostic variables. Discrete entropy stability is closely tied to numerical stability, and many numerical methods achieve stability by ensuring discrete entropy stability \cite{waruszewski2022entropy, ranocha2020entropy, ducros2000high, morinishi2010skew, sjogreen2019entropy, gassner2016split, hennemann2021provably, pirozzoli2010generalized, chan2018discretely, ricardo2024dg, ricardo2024entropy}. For the equilibrium moist compressible Euler equations both tracer variances, $\tfrac{1}{2}\rho q_w^2$ and $\tfrac{1}{2}\rho\eta^2$, are mathematical entropies. We prove conservation of the tracer variances in the following theorem.
\begin{theorem}\label{theo:conservation_of_variance}
Consider the the moist compressible Euler equations \eqref{eq:skew-vel1}-\eqref{eq:skew-water} on the smooth spatial domain $\Omega \subset \mathbb{R}^2$ with periodic and wall \eqref{eq:no-flux-boundary} boundary conditions. For continuous solutions $(\vb{v}, \rho, \eta, q_w): \Omega\times\mathbb{R}^{+} \to \mathbb{R}^{5}$, at time $t\ge 0$ let the total entropy and water variances be denoted by $\mathcal{Y}(t) = \int_\Omega{\tfrac{1}{2}\rho \eta^2d\Omega}$ and $\mathcal{Z}(t) = \int_\Omega{\tfrac{1}{2}\rho q_w^2d\Omega}$. The total entropy and water variances are conserved, that is we have
\begin{align}
    \frac{d \mathcal{Y}}{dt} = 0, \quad \frac{d\mathcal{Z}}{dt} = 0, \quad \forall t\ge 0.
\end{align}
\end{theorem}
\begin{proof}

The proofs for entropy variance and water variance are identical so we consider $\tfrac{1}{2}\rho \gamma^2$ with $\gamma \in \{\eta, q_w\}$,  ${\Upsilon}(t) = \int_\Omega{\tfrac{1}{2}\rho \gamma^2d\Omega}$ and 
$$
\frac{\partial \left(\rho \gamma\right)}{\partial t} + \nabla\cdot \left(\gamma\rho\vb{v}\right) = 0.
$$
Consider the time derivative of ${\Upsilon}(t)$, we have
\begin{equation}
\begin{split}
    \frac{d \Upsilon}{dt} = \frac{\partial }{\partial t}\int_{\Omega}{\tfrac{1}{2}\rho\gamma^2d\Omega} &= \int_{\Omega}{\left(\gamma \frac{\partial \rho \gamma}{\partial t} - \tfrac{1}{2}\gamma^2 \frac{\partial \rho}{\partial t}\right) d\Omega} \\
    &= -\int_{\Omega}{\left(\gamma \nabla \cdot \left(\gamma \vb{F}\right) - \tfrac{1}{2}\gamma^2 \nabla \cdot \vb{F}\right) d\Omega}.
\end{split}
\end{equation}
We apply the chain rule and the product rule, that is
\begin{align}\label{eq:chain-product-rule}
    \nabla \cdot \left(\gamma \vb{F} \right)= \tfrac{1}{2} \nabla \cdot \left(\gamma \vb{F}\right) + \tfrac{1}{2} \vb{F}\cdot \nabla \gamma + \tfrac{1}{2}\gamma\nabla \cdot  \vb{F},
\end{align}
and we have
\begin{equation}
\begin{split}
    \frac{d \Upsilon}{dt} 
    &= -\int_{\Omega}{\left(\tfrac{1}{2}\gamma \nabla \cdot \left(\gamma \vb{F}\right) + \tfrac{1}{2}\gamma \vb{F}\cdot \nabla \gamma + \tfrac{1}{2}\gamma^2\nabla \cdot  \vb{F} - \tfrac{1}{2}\gamma^2 \nabla \cdot \vb{F}\right) d\Omega} \\
    &= -\frac{1}{2}\int_{\Omega}{\left(\gamma \nabla \cdot \left(\gamma \vb{F}\right) + (\gamma \vb{F})\cdot \nabla \gamma\right) d\Omega}.
\end{split}
\end{equation}
We use integration by parts (divergence theorem) to convert the volume integrals to contour integrals around the boundary,  and we have 
\begin{equation}
\begin{split}
    \frac{d \Upsilon}{dt} 
    &= -\frac{1}{2}\oint_{\partial \Omega}{(\gamma^2 \vb{F})\cdot \vb{n} d\Gamma}.
    \label{eq:cont-ent-cons}
\end{split}
\end{equation}
As before, the final step of the proof involves eliminating the contour integral on the boundary by imposing wall boundary conditions $\vb{F}\cdot\vb{n}=0$ and cancellation of boundary terms across periodic boundaries, giving
    \begin{equation}
        \begin{split}
            \frac{d \Upsilon}{dt}  
        &= 0.
        \end{split}
    \end{equation}
    The proof is complete.
    
    \end{proof}

The key ingredients that are necessary in proving the continuous results, Theorem 
\ref{theo:conservation_of_energy}, and Theorem \ref{theo:conservation_of_variance},
are integration by parts,  the chain rule and the product rule. However, while SBP numerical derivative operators  are  designed to mimic integration by parts property at the discrete level, it is challenging  for SBP operators to mimic the chain rule or  product rule at the discrete level.  To be able replicate the continuous analysis at the discrete level, the moist compressible Euler equations \eqref{eq:skew-vel1}-\eqref{eq:skew-water} at the continuous level must be rewritten into the so-called skew-symmetric split-form so that we can negate the use of the chain rule and product rule at the discrete level \cite{ricardo2024entropy,ricardo2024dg}.

\subsection{Skew-symmetric split-form and thermodynamic consistent formulation}
Discrete derivative operators, in the SBP form, cannot discretely mimic the product rule or the chain rule \cite{ricardo2024dg, ricardo2024entropy}. So, to discretely conserve the tracer variances we introduce the required splitting into the advection equations. Simultaneously conserving energy with SBP operators now requires adding corresponding splitting into the skew-symmetric velocity equation \eqref{eq:skew-vel1}. For a vector field $\vb{F}$ and scalar fields $\gamma$, $T$,
we consider the identities
\begin{align}\label{eq:product_rule_1}
     \nabla \cdot \left(\gamma \vb{F}\right) =  \vb{F}\cdot \nabla \gamma + \gamma\nabla \cdot  \vb{F},
\end{align}
\begin{align}\label{eq:product_rule_2}
    \gamma \nabla T = \nabla \left(\gamma T\right) - T \nabla \gamma.
\end{align}
Our skew-symmetric split-form and thermodynamically consistent formulation  of moist compressible Euler equations is obtained by replacing the corresponding terms in \eqref{eq:skew-vel1}-\eqref{eq:skew-water} with the averages of the left and right hand sides of \eqref{eq:product_rule_1}-\eqref{eq:product_rule_2}.  We have
\begin{equation}
\begin{split}
	\frac{\partial \vb{v}}{\partial t} &+ \boldsymbol{\omega}\cross \vb{v} + \nabla G + \tfrac{1}{2}\left(\eta \nabla T + \nabla \left(\eta T\right) - T \nabla \eta\right) \\ &+ \tfrac{1}{2}\left(q_w \nabla \mu_w +  \nabla \left(q_w\mu_w\right) - \mu_w\nabla q_w \right) = 0,
	\label{eq:vel-skew-split}
 \end{split}
\end{equation}
\begin{equation}
	\frac{\partial \rho}{\partial t} + \nabla\cdot \rho\vb{v} = 0,
 \label{eq:den-skew-split}
\end{equation}
\begin{equation}
	\frac{\partial \rho\eta}{\partial t} + \tfrac{1}{2}\left(\nabla\cdot \eta\vb{F} + \eta \nabla\cdot \vb{F} + \vb{F}\cdot \nabla \eta\right) = 0,
 \label{eq:ent-skew-split}
\end{equation}
\begin{equation}
	\frac{\partial \rho q_w}{\partial t} + \tfrac{1}{2}\left(\nabla\cdot q_w\vb{F} + q_w \nabla\cdot \vb{F} + \vb{F}\cdot \nabla q_w\right) = 0.
 \label{eq:wat-skew-split}
\end{equation}
It is significantly noteworthy that, using our skew-symmetric split-form and thermodynamic consistent formulation  of moist compressible Euler equations \eqref{eq:vel-skew-split}--\eqref{eq:wat-skew-split},  we can now prove the continuous results, Theorem 
\ref{theo:conservation_of_energy} and Theorem \ref{theo:conservation_of_variance}, by applying the energy method and
employing integration by parts only, without the chain rule and the product rule.
In section 4, we will  target the skew-symmetric split-form and thermodynamic consistent formulation  of moist compressible Euler equations \eqref{eq:vel-skew-split}--\eqref{eq:wat-skew-split} with our discontinuous Galerkin discretisation. We will prove numerical  stability by emulating the continuous results, Theorem 
\ref{theo:conservation_of_energy} and Theorem \ref{theo:conservation_of_variance}, at the discrete level.
    



\section{Thermodynamics}

This section presents the thermodynamic model we use to close the moist compressible Euler equations. We treat moist air as a mixture of dry air, water vapour, liquid water, and ice, and following \cite{eldred2022thermodynamically} assume:
\begin{itemize}
	\item gaseous components (dry air and water vapour) are ideal gases,
	\item condensates (liquid water and ice) are incompressible and have negligible volume,
	\item constant specific heats
       \item all components are in thermodynamic equilibrium.
\end{itemize}
These are used to derive the thermodynamic potentials, partition the water mass fraction $q_w$ into the separate components $q_v$, $q_l$, and $q_i$, and then calculate the combined water chemical potential $\mu_w$.

\subsection{Potentials}
Here $u^j$, $\eta^j$, $g^j$, $R^j$, $C_v^j$, and $C_p^j$  are the specific internal energy, specific entropy, Gibbs potential, ideal gas constant, and specific heats at constant volume and pressure of the components $j\in{d, v, l, i}$ corresponding to dry air, water vapour, liquid water, and ice. The values for these constants can be found in appendix B. We use internal energy as our thermodynamic potential as in \cite{bowen2022consistent1,bowen2022consistent2}, and with the above approximations these potentials are
\begin{equation}
	u^d = C^d_v \exp\left(\frac{\eta^d}{C^d_v}\right)\left(q^d\rho\right)^{\frac{R^d}{C_v^d}},
 \label{eq:dry-pot}
\end{equation}
\begin{equation}
	u^v = C^v_v T_0\exp\left(\frac{\eta^v}{C^v_v}- \frac{C^v_p}{C^v_v} - \frac{L^s_{00}}{C^v_vT_0}\right)\left(\frac{q^v\rho}{\rho^v_0}\right)^{\frac{R^v}{C_v^v}} + L^s_{00},
 \label{eq:vap-pot}
\end{equation}
\begin{equation}
	u^l = C^l T_0\exp\left(\frac{\eta^l}{C^l}- 1 - \frac{L^f_{00}}{C^lT_0}\right) + L^f_{00},
 \label{eq:liq-pot}
\end{equation}
\begin{equation}
	u^i = C^i T_0\exp\left(\frac{\eta^i}{C^i}- 1\right).
 \label{eq:ice-pot}
\end{equation}

\subsection{Thermodynamic equilibrium}

The thermodynamic potentials, and hence the pressure gradient calculations, require the water species mass fractions. These mass fraction can be inferred from the prognostic variables under the thermodynamic equilibrium assumption. For a given specific entropy $\eta$, density $\rho$, and water mass fraction $q_w$, the second law of thermodynamics implies that the equilibrium values of $q^j$ will minimise the specific internal energy. Mathematically the problem we solve is: given $\rho$, $q^w$, $q^d=1-q_w$, and $\eta = \sum_j{q^j\eta^j}$, minimize $u=\sum_j{q^j u^j}$ subject to $q^j \geq 0$ and $q^w = q^v + q^l + q^i$. Details can be found in appendix A.


%


\section{Spatial discretisation}

Here we present the details of the specific DG discretisation we use. Note that any SBP discretisation of equations \eqref{eq:vel-skew-split}-\eqref{eq:wat-skew-split} with suitable numerical fluxes would achieve the same conservation and stability results.

\subsection{Spaces}
We decompose the domain into quadrilateral elements $\Omega^m$. For each element we define an invertible map to the reference square $r(\vb{x};m): \Omega^m \rightarrow [-1, 1]^2$, and approximate solutions by polynomials of order $p$ in this reference square. We use a computationally efficient tensor product Lagrange polynomial basis with interpolation points collocated with Gauss-Lobatto-Legendre (GLL) quadrature points. Using this basis the polynomials of order $p$ on the reference square can be defined as
\begin{equation}
    P_p := \text{span}_{i, j=1}^{p+1} l_i(\xi) l_j(\eta)\text{,}
\end{equation}
where $(\xi, \eta) \in [-1, 1]^2$ and $l_i$ is the 1D Lagrange polynomial which interpolates the $i^{th}$ GLL node. We define a discontinuous scalar space $S$ as
\begin{equation}
    S = \{\phi \in L^2(\Omega) : \phi\left(r^{-1}\left(\xi, \eta; m\right)\right) \in P_p, \forall m \} \text{.}
\end{equation}
Note that within each element functions $\phi \in S$ can be expressed as
\begin{equation}
    \phi(\vb{x}) = \sum_{i,j=1}^{p+1}{\phi_{ijm}l_i(\xi)l_j(\eta)}, \forall \vb{x} \in \Omega^m\text{,}
\end{equation}
where $\xi, \eta = r(\vb{x}; m)$, $\phi_{ijm} = \phi(r^{-1}(\xi_i, \eta_j;m))$.

We then define a discontinuous vector-valued function space as a tensor product of $S$. Let $\vb{v}_1(x,y,z)$ and $\vb{v}_2(x,y,z)$ be any independent set of vectors which span the tangent space of the manifold at each point $(x,y,z)$, then the discontinuous vector space $V$ containing the velocity $\vb{v}$ can be defined as
\begin{equation}
    V = \{\vb{w} : \vb{w} \cdot \vb{v}_i \in S, i=1,2 \}\text{.}
\end{equation}
We note that $V$ is independent of the particular choice of $\vb{v}_i$ however two convenient choices are the contravariant and covariant basis vectors.

\subsection{Operators}

Before defining our discrete operators we first introduce the covariant vectors
\begin{equation}
    \vb{g}_1 = \dfrac{\partial \vb{x}}{\partial \xi}, \quad \vb{g}_2 = \dfrac{\partial \vb{x}}{\partial \eta}\text{,}
\end{equation}
and the contravariant vectors
\begin{equation}
    \vb{g}^1 = \nabla \xi, \quad \vb{g}_2 = \nabla \eta \text{.}
\end{equation}
With these, the determinant of the Jacobian of $r(x; m)$ is given by
\begin{equation}
    J = |\vb{g}_1 \times \vb{g}_2|\text{.}
\end{equation}

Following \cite{taylor2010compatible}, we use the covariant and contravariant vectors to discretise the divergence, gradient, and curl as
\begin{equation}
    \nabla_d \cdot \vb{w} = \dfrac{1}{J}\bigg(\dfrac{\partial J\vb{w}\cdot \vb{g}^1}{\partial \xi} + \dfrac{\partial J\vb{w}\cdot \vb{g}^2}{\partial \eta} \bigg)\text{,}
    \label{eq:div}
\end{equation}
\begin{equation}
    \nabla_d \phi = \dfrac{\partial \phi}{\partial \xi} \vb{g}^1 + \dfrac{\partial \phi}{\partial \eta} \vb{g}^2\text{,}
    \label{eq:grad}
\end{equation}
\begin{equation}
    \nabla_d \cross \vb{w} = \dfrac{1}{J}\bigg(\dfrac{\partial \vb{w}\cdot\vb{g}_2}{\partial \eta} - \dfrac{\partial \vb{w}\cdot\vb{g}_1}{\partial \xi} \bigg)\vb{k}\text{,}
    \label{eq:curl}.
\end{equation}

\subsection{Integrals}

We approximate integrals over the elements by using GLL \newline quadrature. This defines a discrete element inner product
\begin{equation}
    \left\langle f, g \right\rangle_{\Omega^m} := \sum_{i, j=1}^{p+1}{w_i w_j J_{ij} f_{ij}g_{ij}}\text{,}
\end{equation}
where $w_i$ are the 1D quadrature weights. This approximates the integral
\begin{equation}
    \int_{\Omega^m}{fg d\Omega^m} \approx \left\langle f, g \right\rangle_{\Omega^m}\text{.}
\end{equation}
The discrete global inner product is defined simply as $\left\langle f, g \right\rangle_{\Omega} = \sum_m{\left\langle f, g \right\rangle_{\Omega^m}}$. Similarly for vectors
\begin{equation}
    \left\langle \vb{f}, \vb{g} \right\rangle_{\Omega^m} := \sum_{i, j=1}^{p+1}{w_i w_j J_{ij} \vb{f}_{ij} \cdot \vb{g}_{ij}}\text{.}
\end{equation}
We also approximate element boundary integrals using GLL quadrature
\begin{equation}
    \int_{\partial \Omega^m}{\vb{f}\cdot\vb{n}dl} \approx \left\langle \vb{f}, \vb{n}\right\rangle_{\partial\Omega^m}\text{,}
\end{equation}
where
\begin{equation}
\begin{split}
    \left\langle \vb{f}, \vb{n}\right\rangle_{\partial \Omega^m} := \sum_{j}{w_j\big(|g_1|_{1j} \vb{f}_{1j}\cdot \vb{n}_{1j} + |g_1|_{(p+1)j} \vb{f}_{(p+1)j}\cdot \vb{n}_{(p+1)j}} \\ {+ |g_2|_{j1}\vb{f}_{j1}\cdot \vb{n}_{j1} + |g_2|_{j(p+1)}\vb{f}_{j(p+1)}\cdot \vb{n}_{j(p+1)}\big)}\text{,}
\end{split}
\end{equation}
and $\vb{n}$ is the unit normal of the element boundary $\partial \Omega^m$. Similarly, we use $\vb{t}$ to denote the tangent vectors to the element boundaries, and we also use the boundary mean and jump notation
\begin{equation}
    \average{a} = \tfrac{1}{2}(a^+ + a^-)\text{,}
\end{equation}
\begin{equation}
    \jump{a} = a^+ - a^-\text{,}
\end{equation}
where $+$ and $-$ superscripts are arbitrary and denote quantities on opposite sides of the element boundary.

With the discrete setup defined, we now present the following lemma which is crucial to our discrete stability and conservation proofs.
\begin{lemma}
The discrete spaces satisfy an element-wise SBP property.
\begin{equation}
    \left\langle \phi, \nabla_d\cdot \vb{w}\right\rangle_{\Omega^m} + \left\langle \nabla_d\phi, \vb{w}\right\rangle_{\Omega^m} = \left\langle \phi\vb{w}, \vb{n}\right\rangle_{\partial\Omega^m}\text{,}\; \forall \phi \in S, \vb{w} \in V\text{.}
    \label{eq:sbp}
\end{equation}
\end{lemma}
\begin{proof}
    See \cite{taylor2010compatible}.
\end{proof}

\subsection{DG semi-discrete formulation}

We use the strong form DG-SEM, which has shown to be equivalent to the weak form \cite{kopriva2010quadrature}. Our discrete method is
\begin{equation}
\begin{split}
\left\langle \vb{w}, \vb{v}_t\right\rangle_{\Omega^m} &+ \left\langle \vb{w}, \omega \vb{k}\cross \vb{v} + \nabla_d G\right\rangle_{\Omega^m} + \left\langle \vb{w} \cdot \vb{n},\hat{G}-G\right\rangle_{\partial\Omega^m} \\
&+ \left\langle \vb{w}, \tfrac{1}{2}(q_w\nabla_d \mu_w + \nabla_d \left(\mu_w q_w\right) - \mu_w\nabla_d q_w)\right\rangle_{\Omega^m} + \left\langle \vb{w} \cdot \vb{n},\hat{q_w}\left(\hat{\mu}_w -\mu_w\right)\right\rangle_{\partial\Omega^m} \\
&+\left\langle \vb{w}, \tfrac{1}{2}(\eta\nabla_d T + \nabla_d \left(T \eta\right) - T\nabla_d \eta)\right\rangle_{\Omega^m} + \left\langle \vb{w} \cdot \vb{n},\hat{\eta}\left(\hat{T}-T\right)\right\rangle_{\partial\Omega^m} = 0\text{,}\;\forall \vb{w}\in V\text{,}
\end{split}
\label{eq:dg-velocity}
\end{equation}
\begin{equation}
	\left\langle \phi, \rho_t\right\rangle_{\Omega^m} + \left\langle \phi, \nabla_d \cdot \vb{F}\right\rangle_{\Omega^m} + \left\langle \phi, \big(\hat{\vb{F}}-\vb{F}\big)\cdot \vb{n}\right\rangle_{\partial\Omega^m} = 0\text{,}\;\forall \phi\in S\text{,}
 \label{eq:dg-depth}
\end{equation}
\begin{equation}
\begin{split}
	\left\langle \phi, (\rho\eta)_t\right\rangle_{\Omega^m} &+ \left\langle \phi, \tfrac{1}{2}(\eta\nabla_d \cdot \vb{F} + \vb{F}\cdot \nabla_d \eta + \nabla_d \cdot \left(\eta\vb{F}\right))\right\rangle_{\Omega^m} \\
 & + \left\langle \phi, \big(\hat{\eta}\hat{\vb{F}}-\eta\vb{F}\big)\cdot \vb{n}\right\rangle_{\partial\Omega^m} = 0\text{,}\;\forall \phi\in S\text{,}
 \end{split}
 \label{eq:dg-ent}
\end{equation}
\begin{equation}
\begin{split}
	\left\langle \phi, (\rho q_w)_t\right\rangle_{\Omega^m} &+ \left\langle \phi, \tfrac{1}{2}(q_w\nabla_d \cdot \vb{F} + \vb{F}\cdot \nabla_d q_w + \nabla_d \cdot \left(q_w\vb{F}\right))\right\rangle_{\Omega^m} \\
 & + \left\langle \phi, \big(\hat{q}_w\hat{\vb{F}}-q_w\vb{F}\big)\cdot \vb{n}\right\rangle_{\partial\Omega^m} = 0\text{,}\;\forall \phi\in S\text{,}
 \end{split}
 \label{eq:dg-wat}
\end{equation}
where the discrete vorticity is
\begin{equation}
    \left\langle \phi, \omega\right\rangle_{\Omega^m} = \left\langle  \phi\vb{k}, \nabla_d \cross\vb{v}\right\rangle_{\Omega^m} + \left\langle \phi, f\right\rangle_{\Omega^m} + \left\langle \phi, \left(\hat{\vb{v}} - \vb{v}\right) \cdot \vb{t}\right\rangle_{\partial\Omega^m}\;\forall \phi\in S\text{,}
    \label{eq:dg-vort}
\end{equation}
and the $\hat{\;}$ terms denote numerical fluxes. The numerical fluxes extend those in \cite{ricardo2024dg,ricardo2024conservation}, and are defined for the scalars $\alpha, \beta$ as
\begin{equation}
    \hat{G} = \average{ G } + \alpha\jump{\vb{F}}\cdot{\vb{n}^{+}}\text{.}
\end{equation}
\begin{equation}
    \hat{\vb{F}} = \average{ \vb{F} }\text{,}
\end{equation}
\begin{equation}
    \hat{\vb{v}} = \average{ \vb{v} }\text{,}
\end{equation}
\begin{equation}
    \hat{T} = \average{ T }\text{,}
\end{equation}
\begin{equation}
    \hat{\mu}_w = \average{ \mu_w }\text{,}
\end{equation}
\begin{equation}
    \hat{q}_w = \average{q_w} + \beta \jump{q_w}\text{sign}{\left(\hat{\vb{F}}\cdot\vb{n}^+\right)}\text{,}
\end{equation}
\begin{equation}
    \hat{\eta} = \average{\eta} + \beta \jump{\eta}\text{sign}{\left(\hat{\vb{F}}\cdot\vb{n}^+\right)}\text{.}
    \label{eq:ent-flux}
\end{equation}
We impose no-flux/wall boundary conditions $\vb{F}=0$ at the top and bottom boundaries by setting the numerical fluxes: $\hat{\vb{F}} = 0$, $\hat{G} = G + \gamma \vb{F}\cdot\vb{n}^+$, and $\hat{\vb{v}} = \vb{v}$, $\hat{T} = T$, $\hat{\mu}_w=\mu_w$.

As detailed in the subsequent section, our method is energy conserving for $\alpha = \gamma = 0$ and energy stable for $\alpha \geq 0$, $\gamma \geq 0$. Similarly, it is tracer variance conserving for $\beta = 0$ and tracer variance stable for $\beta \geq 0$. We note that a choice of $\beta=\tfrac{1}{2}$ corresponds to choosing upwind values of $\eta$ and $q_w$.

\section{Semi-discrete conservation and numerical stability}

In this section, we present and prove the conservation and stability properties of our spatial discretization \eqref{eq:dg-velocity}-\eqref{eq:ent-flux}. The main results of this paper are presented in this section, including energy and tracer variance conservation/stability, and conservation of mass, entropy, and water.

The semi-discrete analysis presented here closely mirrors the continuous analysis in Section 2. In the following proofs, we first employ SBP to demonstrate conservation within a single element, following the structure of the continuous proofs. Next, we utilize the numerical fluxes to establish conservation/stability over the whole domain, taking into account inter-element coupling and boundary conditions.

\subsection{Conservation of mass, entropy, and water}

Let $\mathcal{M}^m$, $\mathcal{S}^m$, and $\mathcal{W}^m$ be the total mass, entropy, and water mass in cell $m$, and $\mathcal{M}=\sum_m{\mathcal{M}^m}$, $\mathcal{S}=\sum_m{\mathcal{S}^m}$, and $\mathcal{W}=\sum_m{\mathcal{W}^m}$ be the total mass, entropy, and water mass over the entire domain.

\begin{theorem}
    The semi-discrete method \eqref{eq:dg-velocity}-\eqref{eq:ent-flux} conserves mass, entropy and water, that is $\mathcal{M}_t=\mathcal{S}_t=\mathcal{W}_t = 0$.
\end{theorem}
\begin{proof}
    Substituting $\phi=1$ into \eqref{eq:dg-depth} and using SBP as in \eqref{eq:sbp} yields
    \begin{equation}
    \begin{split}
        \mathcal{M}^m_t &= -\left\langle 1, \nabla_d \cdot \vb{F}\right\rangle_{ \Omega^m} -\left\langle 1, \left(\hat{\vb{F}} - \vb{F}\right)\cdot\vb{n}\right\rangle_{\partial \Omega^m} \\
        &= \left\langle \nabla_d 1,  \vb{F}\right\rangle_{ \Omega^m} -\left\langle 1, \hat{\vb{F}}\cdot\vb{n}\right\rangle_{\partial \Omega^m} \\
        &= -\left\langle 1, \hat{\vb{F}} \cdot\vb{n}\right\rangle_{\partial \Omega^m}.
        \end{split}
    \end{equation}
At the top and bottom boundaries $\hat{\vb{F}} = 0$, and at all other element boundaries contributions from neighbouring elements cancel by the continuity of $\hat{\vb{F}}$ giving
\begin{equation}
    \mathcal{M}_t = -\sum_m{\left\langle 1, \hat{\vb{F}} \cdot\vb{n}\right\rangle_{\partial \Omega^m}} = 0.
\end{equation}

Similarly, substituting $\phi=1$ in \eqref{eq:dg-ent} and applying SBP twice gives
\begin{equation}
\begin{split}
	\mathcal{S}_t &= -\left\langle 1, \tfrac{1}{2}(\eta\nabla_d \cdot \vb{F} + \vb{F}\cdot \nabla_d \eta + \nabla_d \cdot \left(\eta\vb{F}\right))\right\rangle_{\Omega^m} - \left\langle 1, \big(\hat{\eta}\hat{\vb{F}}-\eta\vb{F}\big)\cdot \vb{n}\right\rangle_{\partial\Omega^m} \\
 &= -\tfrac{1}{2}\left\langle \eta, \nabla_d \cdot \vb{F} \right\rangle_{\Omega^m}
 -\tfrac{1}{2}\left\langle \vb{F},  \nabla_d \eta\right\rangle_{\Omega^m}- \left\langle 1, \big(\hat{\eta}\hat{\vb{F}}-\tfrac{1}{2}\eta\vb{F}\big)\cdot \vb{n}\right\rangle_{\partial\Omega^m} \\
  &= - \left\langle 1, \hat{\eta}\hat{\vb{F}}\cdot \vb{n}\right\rangle_{\partial\Omega^m}.
 \end{split}
 \end{equation}
 As before
 \begin{equation}
    \mathcal{S}_t = -\sum_m{\left\langle 1, \hat{\eta}\hat{\vb{F}} \cdot\vb{n}\right\rangle_{\partial \Omega^m}} = 0.
\end{equation}
The same reasoning also shows $\mathcal{W}_t=0$.
\end{proof}

\subsection{Energy stability}

Let $\mathcal{E}^m=\left\langle\tfrac{1}{2}\rho \norm{\vb{v}}^2 + \rho u + \rho \Phi\right\rangle_{\Omega^m}$ be the discrete energy contained in element $m$, and $\mathcal{E}=\sum_m{\mathcal{E}^m}$ be the total discrete energy over the domain $\Omega$. Here we show that our semi-discrete method is energy stable. The result is contained in the following proof.

\begin{theorem}
    The semi-discrete method \eqref{eq:dg-velocity}-\eqref{eq:ent-flux} is energy stable $\mathcal{E}_t \leq 0$ for $\alpha,\gamma \geq 0$, and energy conserving $\mathcal{E}_t = 0$ for $\alpha=\gamma=0$.
\end{theorem}
\begin{proof}
    Following the continuous analysis
    \begin{equation}
        \mathcal{E}^m_t = \left\langle \vb{F}, \vb{v}_t\right\rangle_{\Omega^m}+\left\langle G, \rho_t\right\rangle_{\Omega^m}+\left\langle T, \left(\rho\eta\right)_t\right\rangle_{\Omega^m}+\left\langle \mu_w, \left(\rho q_w\right)_t\right\rangle_{\Omega^m},
    \end{equation}
 where $\boldsymbol{F}$, $G$, $T$, $\mu_w$ are the variational derivatives
        of the energy with respect to $\vb{v}$, $\rho$, $\rho\eta$, $\rho q_w$ as given in \eqref{eq:func-derivs-1}-\eqref{eq:func-derivs-2}.

Substituting $G$, $T$, and $\mu_w$ as test functions into equations \eqref{eq:dg-depth}-\eqref{eq:dg-wat} and applying SBP as in \eqref{eq:sbp}
\begin{equation}
    \left\langle G, \rho_t \right\rangle_{\Omega^m} = \left\langle \nabla_d G, \vb{F} \right\rangle _{\Omega^m}-\left\langle G, \hat{\vb{F}}\cdot \vb{n} \right\rangle _{\partial \Omega^m},
\end{equation}
\begin{equation}
    \left\langle T, \left(\rho\eta\right)_t \right\rangle_{\Omega^m} = \left\langle \vb{F}, \tfrac{1}{2}\left(\eta \nabla_d T + \nabla_d \left(\eta T\right) -T\nabla_d \eta \right) \right\rangle _{\Omega^m}-\left\langle T, \hat{\eta}\hat{\vb{F}}\cdot \vb{n} \right\rangle _{\partial \Omega^m},
\end{equation}
\begin{equation}
\begin{split}
    \left\langle \mu_w, \left(q_w\eta\right)_t \right\rangle_{\Omega^m} &= \left\langle \vb{F}, \tfrac{1}{2}\left(q_w \nabla_d \mu_w + \nabla_d \left(q_w \mu_w\right) -\mu_w\nabla_d q_w \right) \right\rangle _{\Omega^m}\\ &-\left\langle \mu_w, \hat{q}_w\hat{\vb{F}}\cdot \vb{n} \right\rangle _{\partial \Omega^m}.
\end{split}
\end{equation}
Next we substitute $\vb{F}$ as test function in \eqref{eq:dg-velocity}
\begin{equation}
\begin{split}
\left\langle \vb{F}, \vb{v}_t\right\rangle_{\Omega^m} = &-\left\langle \vb{F},  \nabla_d G\right\rangle_{\Omega^m} \\
&-\left\langle \vb{F}, \tfrac{1}{2}(\eta\nabla_d T + \nabla_d \left(T \eta\right) - T\nabla_d \eta)\right\rangle_{\Omega^m}  \\
&- \left\langle \vb{F}, \tfrac{1}{2}(q_w\nabla_d \mu_w + \nabla_d \left(\mu_w q_w\right) - \mu_w\nabla_d q_w)\right\rangle_{\Omega^m}  \\
& - \left\langle \vb{F} \cdot \vb{n},\left(\hat{G}-G\right)+\hat{\eta}\left(\hat{T}-T\right)+\hat{q_w}\left(\hat{\mu}_w -\mu_w\right)\right\rangle_{\partial\Omega^m},
\end{split}
\end{equation}
where we have used $\vb{F}\cdot \left(\omega \vb{k} \cross \vb{v}\right) = 0$ which holds point-wise for collocated methods. Therefore the volume terms cancel and the element energy evolution equation can be reduced to a numerical energy flux through the element boundary
\begin{equation}
\begin{split}
        \mathcal{E}^m_t =& \left\langle \vb{F}, \vb{v}_t\right\rangle_{\Omega^m}+\left\langle G, \rho_t\right\rangle_{\Omega^m}+\left\langle T, \left(\rho\eta\right)_t\right\rangle_{\Omega^m}+\left\langle \mu_w, \left(\rho q_w\right)_t\right\rangle_{\Omega^m}  \\
        =&- \left\langle \left(
        G \hat{\vb{F}}
        + \left(\hat{G} - G\right)\vb{F}
        \right) \cdot \vb{n}\right\rangle_{\partial \Omega^m} \\
        &-\left\langle \hat{\eta}\left(
        T \hat{\vb{F}}
        + \left(\hat{T} - T\right)\vb{F}
        \right) \cdot \vb{n}\right\rangle_{\partial \Omega^m} \\
        & -\left\langle \hat{q}_w\left(
        \mu_w \hat{\vb{F}}
        + \left(\hat{\mu}_w - \mu_w\right)\vb{F}\right) \cdot \vb{n}\right\rangle_{\partial \Omega^m}.
        \end{split}
        \label{eq:energy-dg-integral}
    \end{equation}
    Next, we show that the contributions from the no-flux boundaries and the element interfaces are each $\leq 0$. 
    
    At the no-flux boundaries $\Gamma \subset \partial \Omega$: $\hat{\vb{F}} = 0$, $\hat{T}=T$, $\hat{\mu}_w = \mu_w$, and $\hat{G} = G + \gamma \left(\vb{F}\cdot\vb{n}^+\right)$ therefore the contribution from this boundary to \eqref{eq:energy-dg-integral} becomes $-\left\langle \gamma\left(\vb{F}\cdot \vb{n}^+\right)^2\right\rangle_{\Gamma} \leq 0$.

    Finally, we show that the numerical energy flux jump is $\leq 0$ at the element interfaces. Considering the 3 components of the energy flux separately and using $\jump{ab} = \average{a}\jump{b} + \jump{a}\average{b}$ together with (58) yields
    \begin{equation}
        -\hat{q}_w\jump{\mu_w \hat{\vb{F}} + \hat{\mu}_w\vb{F} - \mu_w\vb{F}}\cdot \vb{n}^+ = -\hat{q}_w\jump{\mu_w \average{F} + \average{\mu_w}\vb{F} - \mu_w\vb{F}}\cdot \vb{n}^+ = 0,
    \end{equation}
    \begin{equation}
        -\hat{\eta}\jump{T \hat{F} + \hat{T}\vb{F} - T\vb{F}}\cdot \vb{n}^+ = 0,
    \end{equation}
    \begin{equation}
        -\jump{G \hat{F} + \hat{G}\vb{F} - G\vb{F}}\cdot \vb{n}^+ = -\alpha \left(\jump{F}\cdot \vb{n}^+\right)^2 \leq 0.
    \end{equation}
    Therefore $\mathcal{E}_t \leq 0$ for $\alpha,\gamma \geq 0$, and $\mathcal{E}_t = 0$ for for $\alpha=\gamma=0$.
\end{proof}

\subsection{Tracer variance stability}
Here we show that the entropy and water variances, 
$\mathcal{Y}= \sum_m{\mathcal{Y}^m}$ and $\mathcal{Z}=\sum_m{\mathcal{Z}^m}$, are stable in our semi-discrete method. As before let $\mathcal{Y}^m=\left\langle\tfrac{1}{2}\rho \eta^2 \right\rangle_{\Omega^m}$ and $\mathcal{Z}^m=\left\langle\tfrac{1}{2}\rho q_w^2 \right\rangle_{\Omega^m}$ be the element entropy and water variances. The result is contained in the following theorem.

\begin{theorem}
    The discrete method \eqref{eq:dg-velocity}-\eqref{eq:ent-flux} is tracer variance stable $\mathcal{Y}_t,\mathcal{Z}_t\leq 0$ when $\beta \geq 0$, and tracer variance conserving $\mathcal{Y}_t=\mathcal{Z}_t=0$ when $\beta = 0$.
\end{theorem}
\begin{proof}
    Following  the continuous analysis \eqref{eq:cont-ent-cons} and applying SBP yields the entropy variance flux
\begin{equation}
\begin{split}
    \mathcal{Y}^m_t &= \left\langle \eta, \left(\rho\eta\right)_t \right\rangle_{\Omega^m} - \left\langle \tfrac{1}{2}\eta^2, \rho_t \right\rangle_{\Omega^m} \\
    &= -\tfrac{1}{2}\left\langle \eta, \nabla_d\cdot\left(\eta\vb{F}\right)\right\rangle_{\Omega^m}
    -\tfrac{1}{2}\left\langle \eta\vb{F}, \nabla_d\eta\right\rangle_{\Omega^m}
    -\left\langle \eta, \left(\hat{\eta}\hat{\vb{F}} -\tfrac{1}{2}\eta\vb{F}-\tfrac{1}{2}\eta\hat{\vb{F}}\right)\cdot \vb{n}\right\rangle_{\partial\Omega^m} \\
    &= -\left\langle \eta, \left(\hat{\eta} -\tfrac{1}{2}\eta\right)\hat{\vb{F}}\cdot \vb{n}\right\rangle_{\partial\Omega^m}.
    \end{split}
    \label{eq:disc-entr-var-integral}
\end{equation}
As before we prove stability by showing that both the domain boundary and element interface conditions are $\leq 0$.

On the domain boundary $\hat{\vb{F}}\cdot \vb{n} = 0$, therefore the domain boundary contribution to \eqref{eq:disc-entr-var-integral} is $0$.

We now show that the entropy variance flux jump is negative semi-definite on the element interfaces. Using $\jump{a^2} = 2\average{a}\jump{a}$ and the definition of $\hat{\eta}$ \eqref{eq:ent-flux}
\begin{equation}
    -\jump{\eta\hat{\eta} - \tfrac{1}{2}\eta^2}\hat{\vb{F}}\cdot\vb{n}^+ = -\beta \jump{\eta}^2\text{sign}\left(\hat{\vb{F}}\cdot\vb{n}^+\right)\hat{\vb{F}}\cdot\vb{n}^+ \leq 0.
\end{equation}
Therefore $\mathcal{Y}_t \leq 0$ when $\beta \geq 0$, and $\mathcal{Y}_t = 0$ when $\beta = 0$. The same analysis shows that $\mathcal{Z}_t \leq 0$ when $\beta \geq 0$, and $\mathcal{Z}_t = 0$ when $\beta = 0$.
\end{proof}

\section{Numerical results}

 In this section we verify our theoretical results with numerical experiments using our DG-SEM discretisation. All experiments use 3rd order polynomials and a strong stablity preserving (SSP) RK3 time integrator \cite{shu1988efficient}. We also limit the water mass fraction with the positivity preserving limiter introduced in \cite{zhang2010maximum,liu1996nonoscillatory}. This limiter is high order accurate and does not affect the accuracy of our scheme, but it may introduce energy and water variance conservation errors.

In the following test cases, we  start with hydrostatic profiles (specified below) and apply the perturbations of the form
\begin{equation}
    \rho' = \begin{cases}
        \Delta \rho \cos{\left(\pi\frac{r} {2R}\right)}, \; for\ r \leq R \\
        0, \; r > R
    \end{cases}
\end{equation}
where $R$ is the bubble radius, $\Delta \rho$ is the perturbation strength, \newline$r = \sqrt{(x-x_c)^2 + (y-y_c)^2}$, and $(x_c, y_c)$ is the bubble centre.

\subsection{Test case 1: two phase bubble in a neutral atmosphere}

In this first test case, we verify our method by simulating the two phase (water vapour and liquid) bubble configuration described in \cite{bryan2002benchmark} and \cite{thuburn2017use}. This test simulates a neutrally buoyant profile in a $10$ km $\times 10$ km domain perturbed by a warm bubble with parameters: $\Delta \rho = \tfrac{1}{150}$, $R=2$ km, and $(x_c, y_c)=$ ($0$ km, $2$ km). The atmospheric profile in this test case has a constant water mass ratio $q_w=0.02$, a surface pressure of $1000$ \text{hPa}, a surface density of $1.2$ $\text{kgm}^{-3}$, and a constant specific entropy calculated from the surface values $\eta=2540.6$ K.

We present results using: $128^2$ 3rd order elements, a timestep of 0.03s, kinetic-energy dissipation and upwinding $\alpha=\gamma=\beta=0.5$, and the ice phase turned off. Figures \ref{fig:bf-entropy} and \ref{fig:bf-vapour} show the specific entropy and water mass ratios at times 0s, 800s, 1000s, and 1200s. Previous studies \cite{bryan2002benchmark,thuburn2017use,bowen2022consistent1}  verify their results by comparing plots of the solution at 1000s. The height of our bubble is in agreement with these previous studies, but our simulation develops Rayleigh-Taylor instabilities at the leading and trailing edges of the bubble which are absent from other works. We hypothesize that this is because the previous studies include artificial viscosity which is absent from our method, enabling our simulation to develop fine scale features.
\begin{figure}[!hbtp]
\begin{center}
	\includegraphics[width=0.8\textwidth]{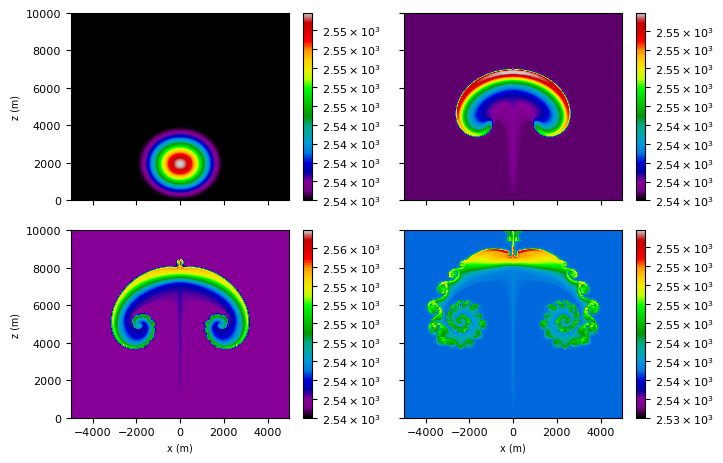}
 \caption{Specific entropy of the Bryan-Fritsch test case at $t=0\text{s},800\text{s},1000\text{s},1200\text{s}$. }
  \label{fig:bf-entropy}
\end{center}
\end{figure}

\begin{figure}[!hbtp]
\begin{center}
	\includegraphics[width=0.8\textwidth]{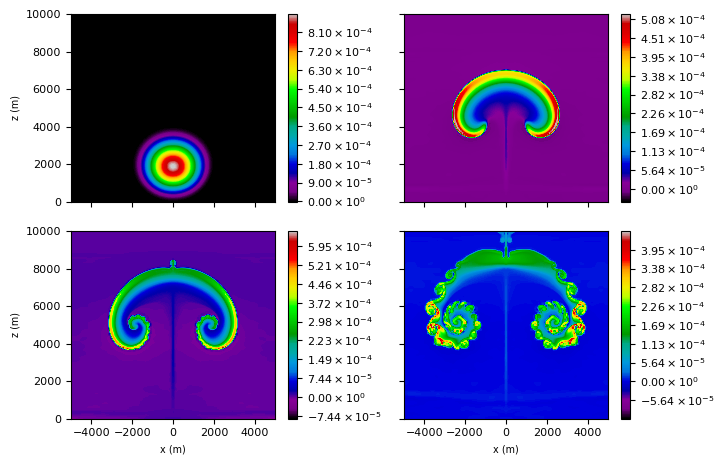}
 \caption{Water vapour mass fraction of the Bryan-Fritsch test case at $t=0\text{s},800\text{s},1000\text{s},1200\text{s}$.}
  \label{fig:bf-vapour}
\end{center}
\end{figure}



\subsection{Test 2: three phase bubble in a an unstable atmosphere}

In the second test case, we verify our full three phase configuration by simulating an unstable sub-saturated background profile perturbed by the same bubble as the first test case. The density and pressure profiles are obtained from a neutrally stable dry profile with a constant dry potential temperature of $300K$ and a surface pressure of $1000hPa$. Given the pressure and density profiles, we then set the water profile by specifying a constant relative humidity of 0.95, therefore no ice or liquid should be present at the beginning of the simulation.

We present results using  $64^2$ 3rd order elements, a timestep of 0.06s, kinetic-energy dissipation and upwinding $\alpha=\gamma=\beta=0.5$, and all three phases. Figure \ref{fig:ice-entropy} shows the entropy perturbation with respect to the background profile. As expected the bubble rises faster in this unstable  profile compared to the neutrally-stable profile in the first test case. Figure \ref{fig:ice-ice} shows the ice mass fraction which is initially zero as the air is sub-saturated, the  bubble then begins to rise transporting moist air upwards and causing ice to form. Importantly, no ice is observed in the lower part of the domain where the temperature is above freezing. Figure \ref{fig:ice-conservation} shows the relative errors of energy, entropy variance, and water variance, verifying that the stability of our method.
\begin{figure}[!hbtp]
\begin{center}
	\includegraphics[width=0.8\textwidth]{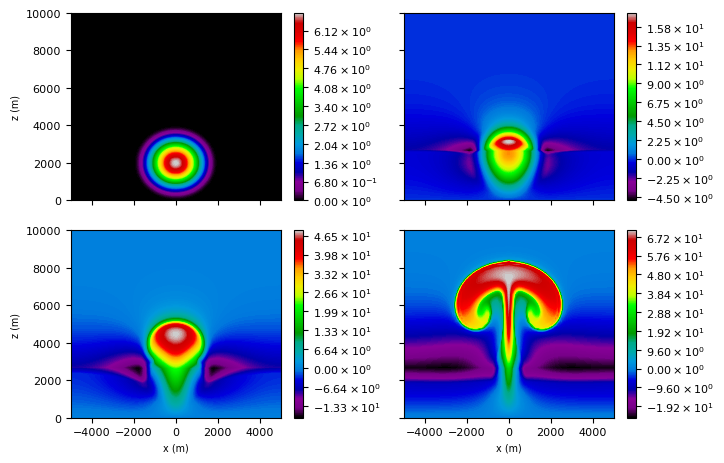}
 \caption{Specific entropy perturbation of the three phase bubble test case at $t=0\text{s},200\text{s},400\text{s},600\text{s}$.}
  \label{fig:ice-entropy}
\end{center}
\end{figure}
\begin{figure}[!hbtp]
\begin{center}
	\includegraphics[width=0.8\textwidth]{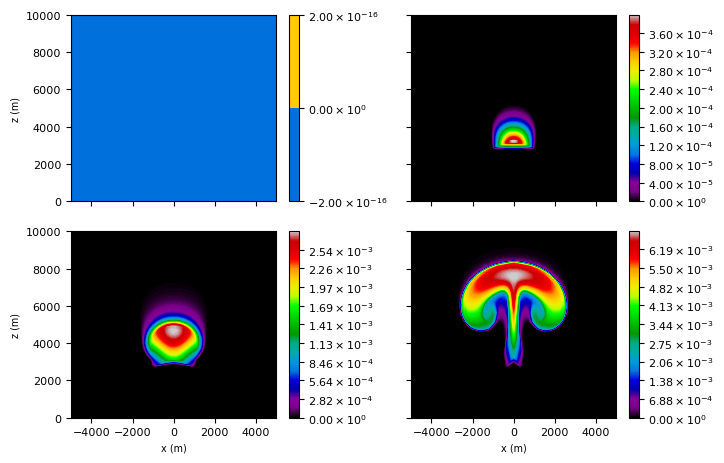}
 \caption{Ice mass fraction of the three phase bubble test case at $t=0\text{s},200\text{s},400\text{s},600\text{s}$.}
  \label{fig:ice-ice}
\end{center}
\end{figure}
\begin{figure}[!hbtp]
\begin{center}
	\includegraphics[width=0.8\textwidth]{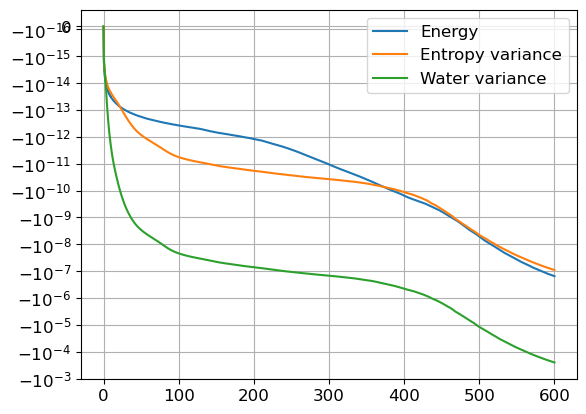}
 \caption{Conservation errors of the three phase bubble test case.}
  \label{fig:ice-conservation}
\end{center}
\end{figure}

\subsection{Conservation}

We verify the semi-discrete conservation properties of our method by running a low resolution version of test case 2 with $8^2$ 3rd order elements and decreasing timesteps. Figure \ref{fig:ice-conservation-convergence} shows the conservation errors. The energy and water variance both decrease at approximately 3rd order, verifying our theoretical results. The entropy variance initially decreases but plateaus at a relative error of $ 10^{-13}$, we hypothesize that this is due to round-off error.
\begin{figure}[!hbtp]\begin{center}
	\includegraphics[width=0.8\textwidth]{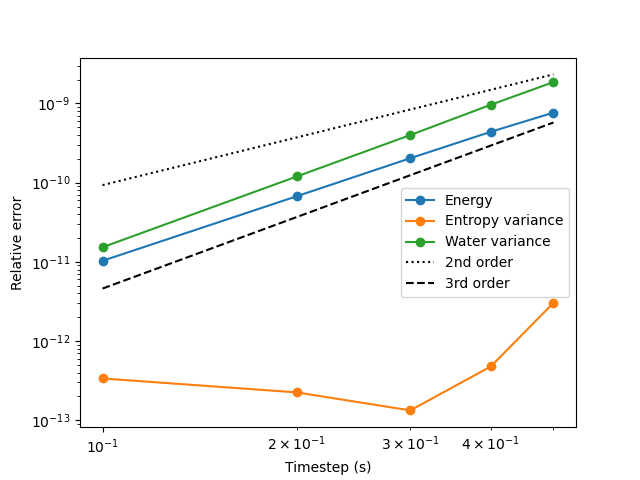}
 \caption{Semi-discrete conservation convergence test of the three phase bubble test at $t=600\text{s}$.}
  \label{fig:ice-conservation-convergence}
\end{center}
\end{figure}

\subsection{Convergence}

To verify the convergence of the numerical method we modify the first test case to use a smooth bubble perturbation of the form $\rho' = \Delta \rho \exp{-\left(\frac{2r} {R}\right)^2}$ and run the simulation until $t=400\text{s}$, stopping the simulation shortly before a shock develops. Figure \ref{fig:convergence} shows the normalised $L^2$ errors compared to a high resolution reference simulation with $128\times 128$ 3rd order elements, showing that our method converges at the expected 3rd order accuracy for the 3rd order time integrator used.

\begin{figure}[!hbtp]\begin{center}
	\includegraphics[width=0.8\textwidth]{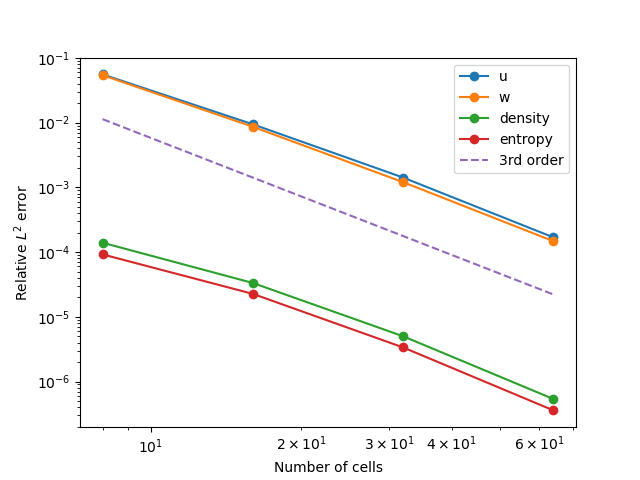}
 \caption{Convergence of a smooth two-phase bubble at $t=400\text{s}$ compared to a high resolution reference solution.} 
  \label{fig:convergence}
\end{center}
\end{figure}

Importantly, in this test case the air is saturated and is everywhere a vapour-liquid mixture. Therefore there are no transitions between vapour only, vapour-liquid, vapour-ice, and vapour-liquid ice states which introduce discontinuities in the temperature and chemical potential gradients and reduce the order of accuracy of the model. 

\subsection{Stability}

To demonstrate the importance of discrete tracer variance stability for ensuring numerical stability, we run test case 1 with $\Delta \rho = \frac{1}{5}$ using two different configurations: our standard DG method \eqref{eq:dg-velocity}-\eqref{eq:ent-flux} (Method 1), and an alternate method (Method 2) that is energy stable but not tracer variance stable. Method 2 is identical to our method except that it excludes the term splitting in the volume integrals of \eqref{eq:dg-velocity}, \eqref{eq:dg-ent}, and \eqref{eq:dg-wat}.

We run both simulations with a resolution of $8\times 8$ elements, no upwinding or kinetic energy dissipation at the element interfaces $\alpha=\beta=0$, and with a boundary condition penalty of $\gamma=0.5$. Figure \ref{fig:stability} shows the energy and water variance conservation errors for the two methods. While both methods are energy stable with similar energy conservation errors, they have distinctly different water variance conservation errors. At around $600\text{s}$  the water variance of Method 2 begins increasing exponentially until the model crashes at approximately $1000{s}$. In contrast our method (Method 1) has a stable water variance and runs stably for the whole duration, demonstrating the increased numerical ability associated with discrete tracer variance stability.

\begin{figure}[!hbtp]\begin{center}
	\includegraphics[width=0.8\textwidth]{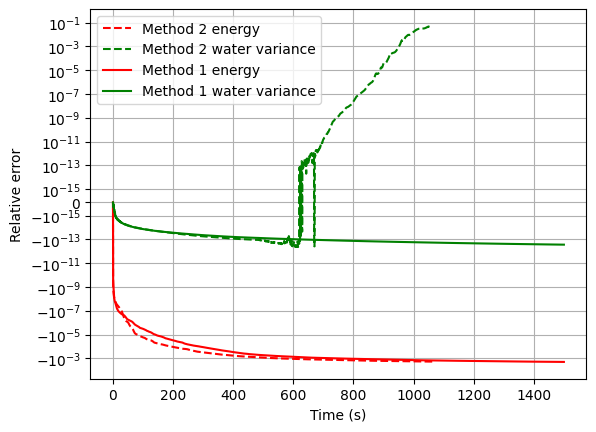}
 \caption{Conservation errors of our method (Method 1) compared to an energy stable but not tracer variance stable method (Method 2).} 
  \label{fig:stability}
\end{center}
\end{figure}

\section{Conclusion}

In this work we developed a general framework for constructing numerical methods of the moist compressible Euler equations with excellent conservation and stability properties using only SBP operators. We then constructed a novel DG-SEM discretisation for the equilibrium moist compressible Euler equations using this framework and proved that our method discretely conserves mass, water, and entropy; and is semi-discretely energy and tracer variance stable/conserving. These theoretical results were then verified with numerical experiments.

To the best of our knowledge, this is the first energy conserving/stable discretisation of the moist compressible Euler equations. Additionally, our approach is tracer variance stable, enabling our method to be numerically stable without applying artificial viscosity. We achieve these results by discretising the operator-split skew-symmetric formulation \cite{ricardo2024dg,ricardo2024entropy} with SBP operators and ensuring thermodynamic consistency through the use of the internal energy potential.

This approach has the potential to improve the energy budget of global climate models and therefore more faithfully model climate dynamics over long time scales. Future work will focus on including this approach in global atmospheric model with non-equilibrium physics.

\section*{Declaration of competing interest}

The authors declare that they have no known competing financial interests or personal relationships that could have appeared to influence the work reported in this paper.

\section*{Data availability}

The code to generate, analyze, and plot the datasets used in this study can be found in the associated reproducibility repository https://github.com/kieranricardo/moist-euler-dg.

\clearpage
\bibliographystyle{unsrtnat}
\bibliography{main}  

\appendix

\section{Thermodynamic equilibrium}

This section presents the approach used to solve for the equilibrium water mass fractions. First, the entropy and Gibbs functions are formulated as functions of temperature, density, and mass fractions. Next, we show how given entropy, density and the water mass fraction, the temperature and internal energy can be calculated. Finally, our optimisation method for finding the equilibrium mass fractions is presented.

\subsection{Entropy and Gibbs functions}

For each component the temperature can be derived from the internal energy components \eqref{eq:dry-pot}-\eqref{eq:ice-pot} through the relationship $T = \frac{\delta u^j}{\delta \eta^j}$. Re-arranging gives entropy as a function of temperature, density, and mass fractions:
\begin{equation}
    \eta^d = C_v^d\log\left(T\right) - R^d\log\left(q^d\rho\right),
    \label{eq:dry-ent}
\end{equation}
\begin{equation}
	\eta^v = C_p^v + \frac{L^s_{00}}{T_0} + C_v^v\log\left(\frac{T}{T_0}\right) - R^v\log\left(\frac{q^v\rho}{\rho^v_0}\right),
 \label{eq:vap-ent}
\end{equation}
\begin{equation}
	\eta^l = C^l + \frac{L^f_{00}}{T_0} + C^l\log\left(\frac{T}{T_0}\right),
    \label{eq:liq-ent}
\end{equation}
\begin{equation}
	\eta^i = C^i + C^i\log\left(\frac{T}{T_0}\right).
    \label{eq:ice-ent}
\end{equation}

The Gibbs functions, defined as $g^j :=  u^j + \frac{p^j}{q^j \rho} - \eta^j T$, can then be expressed as
\begin{equation}
    g^v = u^v + \frac{p^v}{q^v\rho}-\eta^vT=-C_v^vT\log\left(\frac{T}{T_0}\right) + R^vT\log\left(\frac{q^v\rho}{\rho^v_0}\right) + L^s_{00} \left(1 - \frac{T}{T_0}\right),
\end{equation}
\begin{equation}
	g^l = u^l -\eta^lT=-C^lT\log\left(\frac{T}{T_0}\right)  + L^f_{00} \left(1 - \frac{T}{T_0}\right),
\end{equation}
\begin{equation}
	g^i = u^i -\eta^iT=-C^iT\log\left(\frac{T}{T_0}\right).
\end{equation}

\subsection{Combined thermodynamic potential}
In thermodynamic equilibrium all components have the same temperature $T$. Using \eqref{eq:dry-ent}-\eqref{eq:ice-ent} the total entropy can then be expressed as
\begin{equation}
    \eta = \sum_j{q^j \eta^j} = C^*_v \log{\left(T\right)} - R^* \log{\left(\rho\right)} -\sum_{j\in\{d, v\}}{q^j R^j\log{\left(q^j\right)}} + \sum_{j}{q^j \eta^j_r},
\end{equation}

where $C_v^*=\sum_j{q^jC^j_v}$, $R^* = q^d R^d + q^v R^v$, and reference specific entropies $\eta^j_r$ are given by $\eta^j_r = \eta^j \left(T=\rho=q^j = 1\right)$. This enables us to solve for $T$, and hence $u=C_v^*T + q^v L^s_{00} + q^l L^f_{00}$, as a function of total density, specific entropy, and water mass fractions
\begin{equation}
\begin{split}
    u(\rho, \eta, q^d, q^v, q^l, q^i) &=  C_v^* \left[\exp\left(\eta-\sum_{j}{q^j\eta^j_r}\right)
	\rho^{R^*}\left(q^d\right)^{R^dq^d}\left(q^v\right)^{R^vq^v}
	\right]^{\frac{1}{C_v^*}} \\ &+ q^vL^s_{00} + q^lL^f_{00}.
 \end{split}
\end{equation}
We also note that the partial derivatives of $u$ w.r.t. to the water mass fractions is simply the Gibbs function of that component
\begin{equation}
    \frac{\delta u}{\delta q^j} = g^j.
\end{equation}


\subsection{Optimisation}

Here we present our method for finding the equilibrium mass fractions by minimising the specific internal energy. Specifically, we solve the following convex problem:
\begin{equation}
    \text{argmin}_{q^l, q^i}\; u(\rho, \eta, q^w-q^l-q^i, q^l, q^i),
\end{equation}
subject to 
\begin{equation}
    q^l, q^i \geq 0,\quad q^w - q^l - q^i \geq 0,
\end{equation}
where we have eliminated $q^v$ through the constraint $q^w = q^v + q^l + q^i$. The partial derivatives of $u$ are then
\begin{equation}
    \frac{\partial u}{\partial q^l} = g^l - g^v,
\end{equation}
\begin{equation}
    \frac{\partial u}{\partial q^i} = g^i - g^v.
\end{equation}

We then solve this problem by:
\begin{enumerate}
    \item Check whether the triple point $q^v, q^l, q^i >0$ is the solution. This is algebraically simple to verify and presented below.
    \item Check whether vapour only $q^l=q^i=0$ is the solution.
    \item Newton solve for $g^l - g^v = 0$ assuming $q^i=0$. Check whether this is the solution.
    \item Newton solve for $g^i - g^v = 0$ assuming $q^l=0$. Check whether this is the solution.
\end{enumerate}
The $q^w - q^l - q^i \geq 0$ constraint is enforced by limiting $q^v$ to a minimum value of $10^{-15}$ at each Newton iteration. Limiting $q^v$ may prevent the true local minima from being found and therefore introduce energy conservation errors, but the small limiting value ensures that this error is negligible. We did not find it necessary to explicitly enforce the $q^l, q^i \geq 0$ constraints.

\subsubsection{Triple point}

Water and ice co-existence $\implies g^i = g^l = 0 \implies T = T_0 = 0^\circ C$. If water vapour is present it implies $g^v = 0 \implies q^v = \frac{\rho^v_0}{\rho}$. $q^i$ and $q^l$ can then be solved from the entropy equations A.3-A.4 and $q_v = q^v + q^l + q^i$.









\section{Constants}
Here we use the same constants as \cite{bowen2022consistent1}.
\begin{table}[H]
\begin{tabular}{|l|l|l|}
\hline
Constant   & Description                                                  & Value                       \\ \hline
$R^d$      & Gas constant of dry air                                      & $287.0 \text{\;Jkg}^{-1}\text{K}^{-1}$      \\ \hline
$R^v$      & Gas constant of water vapour                                 & $461.0 \text{\;Jkg}^{-1}\text{K}^{-1}$      \\ \hline
$C^d_v$    & Specific heat capacity of dry air, at constant volume        & $717.0 \text{\;Jkg}^{-1}\text{K}^{-1}$      \\ \hline
$C^v_v$    & Specific heat capacity of water vapour, at constant volume   & $1424.0 \text{\;Jkg}^{-1}\text{K}^{-1}$     \\ \hline
$C^l$      & Specific heat capacity of liquid water                       & $4186.0 \text{\;Jkg}^{-1}\text{K}^{-1}$     \\ \hline
$C^I$      & Specific heat capacity of ice                                & $2106.0 \text{\;Jkg}^{-1}\text{K}^{-1}$     \\ \hline
$L^s_{00}$ & Specific latent heat of sublimation at $T=p=0$               & $2773633.85 \text{\;Jkg}^{-1}\text{K}^{-1}$ \\ \hline
$L^f_{00}$ & Specific latent heat of freezing at $T=p=0$                  & $902152.0 \text{\;Jkg}^{-1}\text{K}^{-1}$   \\ \hline
$\rho^v_0$ & Reference water vapour density                               & $0.004853 \text{\;kgm}^{-3}$         \\ \hline
$T_0$      & Reference temperature                                        & $273.15 \text{K}$                  \\ \hline
$C_p^v$    & Specific heat capacity of water vapour, at constant pressure & $C^v_v + R^v$               \\ \hline
\end{tabular}
\end{table}

\end{document}